\documentclass[11pt,letterpaper]{article}

\usepackage{amsmath,amsfonts,amsthm,amssymb,stmaryrd,graphicx,relsize,bm,cite}
\theoremstyle{plain}

\numberwithin{equation}{section}

\newtheorem{thm}{Theorem}[section]
\newtheorem{lem}[thm]{Lemma}
\newtheorem{cor}[thm]{Corollary}

\newenvironment{exam}[1]
{\begin{flushleft}\textbf{Example #1}.\enspace}%
{\end{flushleft}}

{\end{flushleft}}

\allowdisplaybreaks  

\newcommand{\complex}{{\mathbb C}}

\newcommand{\trace}{\mathrm{tr}}

\newcommand{\instr}{\mathrm{In}}
\newcommand{\rmco}{\mathop{\rm co}}
\newcommand{\rmt}{\mathop{\rm T}}
\newcommand{\ob}{\mathrm{Ob}}
\newcommand{\ityes}{\textit{yes}}
\newcommand{\itno}{\textit{no}}

\newcommand{\escript}{\mathcal{E}}
\newcommand{\hscript}{\mathcal{H}}
\newcommand{\iscript}{\mathcal{I}}
\newcommand{\jscript}{\mathcal{J}}
\newcommand{\kscript}{\mathcal{K}}
\newcommand{\lscript}{\mathcal{L}}
\newcommand{\mscript}{\mathcal{M}}
\newcommand{\oscript}{\mathcal{O}}
\newcommand{\pscript}{\mathcal{P}}
\newcommand{\sscript}{\mathcal{S}}

\newcommand{\iscripthat}{\widehat{\iscript}}
\newcommand{\jscripthat}{\widehat{\jscript}}
\newcommand{\kscripthat}{\widehat{\kscript}}

\newcommand{\mscripthat}{\widehat{\mscript}}

\newcommand{\hscriptbar}{\overline{\hscript}}
\newcommand{\iscriptbar}{\overline{\iscript}}
\newcommand{\jscriptbar}{\overline{\jscript}}
\newcommand{\kscriptbar}{\overline{\kscript}}

\newcommand{\doubleab}[1]{\left|\left|#1\right|\right|}
\newcommand{\brac}[1]{\left\{#1\right\}}
\newcommand{\paren}[1]{\left(#1\right)}
\newcommand{\sqbrac}[1]{\left[#1\right]}
\newcommand{\elbows}[1]{{\left\langle#1\right\rangle}}
\newcommand{\ket}[1]{{\left|#1\right>}}
\newcommand{\bra}[1]{{\left<#1\right|}}

\begin{document}

\title{A THEORY OF QUANTUM INSTRUMENTS}
\author{Stan Gudder\\ Department of Mathematics\\
University of Denver\\ Denver, Colorado 80208\\
sgudder@du.edu}
\date{}
\maketitle

\begin{abstract}
Until recently, a quantum instrument was defined to be a completely positive operation-valued measure from the set of states on a Hilbert space to itself. In the last few years, this definition has been generalized to such measures between sets of states from different Hilbert spaces called the input and output Hilbert spaces. This article
presents a theory of such instruments.Ways that instruments can be combined such as convex combinations, post-processing, sequential products, tensor products and conditioning are studied. We also consider marginal, reduced instruments and how these are used to define coexistence (compatibility) of instruments. Finally, we present a brief introduction to quantum measurement models where the generalization of instruments is essential. Many of the concepts of the theory are illustrated by examples. In particular, we discuss Holevo and Kraus instruments.
\end{abstract}

\section{Introduction}  
In classical physics a measurement of a physical system does not alter the state of the system. Because of this, a measurement does not interfere with later measurements. An important characteristic of quantum mechanics is that the state of a system can change into an updated state when a measurement is performed. An even more surprising and radical possibility has been recently introduced \cite{dapt22,bkmpt22,mt122,mt222}. These works have pointed out that when the initial state $\rho$ of a quantum system is represented by a density operator on an input Hilbert space $H$, then the updated state after a measurement is performed may be represented by a density operator $\rho _1$ in a different output Hilbert space $H_1$. Not only can the state of the system change as the result of a measurement, but the entire system can be altered so it is described by a different Hilbert space. This is truly an amazing new possibility! In this work we represent measurements by instruments acting on states of a Hilbert space. We present a theory of quantum instruments that emphasizes this new possibility.

Ways that instruments can be combined such as convex combinations, post-processing, tensor products, sequential products and conditioning are studied
\cite{gn01,gud120,gud220,gud21,gud22}. We also consider marginal and reduced instruments. These concepts are employed to define coexistence (compatibility) of instruments and observables. Although compatibility has been well presented in the literature \cite{dapt22,bkmpt22,ls22,mt122,mt222}, we point out some of its features here. Even when two instruments have different output spaces, if their input space $H$ is the same, then the observables they measure are on $H$. Because of this, we can compare these measured observables. Finally, we consider measurement models that can be used to measure instruments \cite{hz12,nc00}. These models strongly rely on the fact that instruments can have different input and output spaces. Many of the concepts of the theory are illustrated by examples. In particular, a theory of Holevo and Kraus instruments are considered \cite{hol82,kra83}.

Section~2 presents the basic concepts and definitions of the theory. In particular, we discuss the concepts of effects, observables, operations and instruments
\cite{bgl95,dl70,gn01,hz12,nc00}. Section 3 gives examples of various instruments that illustrate the theory. An important role is played by Holevo and Kraus instruments \cite{hol82,kra83}. In Section 4, we discuss theorems and results concerning instruments and observables. For example, we show that an observable conditioned on an instrument coexists with the observable measured by the instrument. Section 5 introduces the concept of a quantum measurement model. The instrument that such a model measures employs a L\"uders instrument \cite{lud51}. We also give a new definition of the sequential product of measurement models \cite{gn01}.

\section{Basic Definitions and Concepts}  
In this work, all of our Hilbert spaces are assumed to be finite dimensional. Although this is a strong restriction, it is general enough to include theories of quantum computation and information \cite{hz12,nc00}. We retain this restriction for mathematical simplicity even though many of our results can be extended to the infinite dimensional case. The set of (bounded) linear operators on a Hilbert space $H$ is denoted by $\lscript (H)$ and the zero and identity operators are $0$ and $I$, respectively. When it is necessary to distinguish the Hilbert space, we write $I_H$ instead of $I$. An \textit{operation} from $H$ to $H_1$ is a completely positive, trace non-increasing, linear map $\jscript\colon\lscript (H)\to\lscript (H_1)$ \cite{dl70,hz12,nc00}. We denote the set of operations from $H$ to $H_1$ by $\oscript (H,H_1)$. For simplicity, we write $\oscript (H)=\oscript (H,H)$ when $H=H_1$. If $\jscript _1\in\oscript (H,H_1)$, $\jscript _2\in\oscript (H_1,H_2)$, their \textit{sequential product}
$\jscript _1\circ\jscript _2\in\oscript (H,H_2)$ is given by $\jscript _1\circ\jscript _2(A)=\jscript _2\paren{\jscript _1(A)}$. If $\jscript\in\oscript (H,H_1)$ is trace preserving we call $\jscript$ a \textit{channel}. Every operation $\jscript\in\oscript (H,H_1)$ has the form $\jscript (A)=\sum\limits _{i=1}^nJ_iAJ_i^*$ where $J_i\colon H\to H_1$ is a linear operator with adjoint $J_i^*$ and $\sum\limits _{i=i} ^nJ_i^*J_i\le I_H$ \cite{hz12,nc00}. The operators $J_i$, $i=1,2,\ldots ,n$ are called \textit{Kraus operators} for $\jscript$ \cite{kra83}. We have that $\jscript$ is a channel if and only if $\sum _{i=1}^nJ_i^*J_i=I_H$. If $\jscript\in\oscript (H,H_1)$ we define the unique
\textit{dual map} $\jscript ^*\colon\lscript (H_1)\to\lscript (H)$ by $\trace\sqbrac{B\jscript ^*(A)}=\trace\sqbrac{\jscript (B)A}$ for all $B\in\lscript (H)$, $A\in\lscript (H_1)$
\cite{gud22}. If $\jscript$ has Kraus decomposition $\jscript (A)=\sum\limits _{i=1}^nJ_iAJ_i^*$ then $\jscript ^*(B)=\sum _{i=1}^nJ_i^*BJ_i$. If $\jscript$ is a channel, then $\jscript ^*(I_{H_1})=I_H$ because 
\begin{equation*}
\trace\sqbrac{B\jscript ^*(I_{H_1})}=\trace\sqbrac{\jscript (B)I_{H_1}}=\trace\sqbrac{\jscript (B)}=1=\trace (BI_{H})
\end{equation*}
for all $B\in\lscript (H)$. A positive operator $\rho\in\lscript (H)$ with trace $\trace (\rho )=1$ is called a \textit{state} on $H$. A state describes the condition of a quantum system and the set of states on $H$ is denoted by $\sscript (H)$. We see that if $\rho\in\sscript (H)$ and $\jscript\in\oscript (H,H_1)$ is a channel, then
$\jscript (\rho )\in\sscript (H_1)$. Also, it is easy to check that $(\jscript _1\circ\jscript _2)^*=\jscript _2^*\circ\jscript _1^*$.

A (finite) \textit{instrument} is a finite set $\iscript =\brac{\iscript _x\colon x\in\Omega _\iscript}$ where $\iscript _x\in\oscript (H,H_1)$ such that
$\iscriptbar =\sum\limits _{x\in\Omega _\iscript}\iscript _x$ is a channel \cite{dl70,hz12,nc00}. An instrument is sometimes called an \textit{operation-valued measure}. We call $\Omega _\iscript$ the \textit{outcome space} for $\iscript$ and designate the set of instruments from $H$ to $H_1$ by $\instr (H,H_1)$. We think of
$\iscript\in\instr (H,H_1)$ as an apparatus or experiment that has outcomes $x\in\Omega _\iscript$. The probability that outcome $x$ occurs when $\iscript$ is measured and the system is in state $\rho\in\sscript (H)$ is given by the Born rule $\trace\sqbrac{\iscript _x(\rho )}$ \cite{hz12,nc00}. Since $\iscript _x$ is positive and
$\iscriptbar$ is a channel, we have that $0\le\trace\sqbrac{\iscript _x(\rho )}\le 1$ and $\sum\limits _{x\in\Omega _\iscript}\trace\sqbrac{\iscript _x(\rho )}=1$ so
$x\mapsto\trace\sqbrac{\iscript _x(\rho )}$ is a probability measure on $\Omega _\iscript$. If $\trace\sqbrac{\iscript _x(\rho )}\ne 0$ and $\rho\in\sscript (H)$ is the initial state of the system, then $\iscript _x(\rho )/\trace\sqbrac{\iscript _x(\rho )}\in\sscript (H_1)$ is the \textit{updated} state after the outcome $x$ occurs. As pointed out in Section~1, this updated state can be in a different Hilbert space $H_1$ than the input space $H$. If $\iscript\in\instr (H,H_1)$ we call the probability measure
$\Phi _\rho ^\iscript (x)=\trace\sqbrac{\iscript _x(\rho )}$ the $\rho$-\textit{distribution} of $\iscript$. As we shall see, two different instruments can have the same
$\rho$-distribution for all $\rho\in\sscript (H)$. A \textit{bi-instrument} $\iscript\in\instr (H,H_1)$ is an instrument whose outcome space has the product form
$\Omega _\iscript =\Omega _1\times\Omega _2$ and we write $\iscript _{xy}(\rho )$, $x\in\Omega _1$, $y\in\Omega _2$. In this case, we define the 1-\textit{marginal} and 2-\textit{marginal} of $\iscript$ by $\iscript _x^1(\rho )=\sum\limits _{y\in\Omega _2}\iscript _{xy}(\rho )$ and
$\iscript _y^2=\sum\limits _{x\in\Omega _1}\iscript _{xy}(\rho )$, respectively. This gives us the three instruments $\iscript ,\iscript ^1,\iscript ^2\in\instr (H,H_1)$. Notice that these instruments give the same channels because
\begin{equation*}
\iscriptbar (\rho )=\sum\limits _{xy}\iscript _{xy}(\rho )=\sum\limits _x\sum\limits _y\iscript _{xy}(\rho )=\sum\limits _x\iscript _x^1(\rho )=\iscriptbar\,^1(\rho )
\end{equation*}
and similarly, $\iscriptbar (\rho )=\iscriptbar\,^2(\rho )$ for all $\rho\in\sscript (H)$.

If $\iscript\in\instr (H,H_1)$ and $\jscript\in\instr (H_1,H_2)$, the \textit{sequential product of} $\iscript$ \textit{then} $\jscript$ is the bi-instrument $\iscript\circ\jscript\in\instr (H,H_2)$ given by
\begin{equation*}
(\iscript\circ\jscript )_{xy}(\rho )=\jscript _y\paren{\iscript _x(\rho )}
\end{equation*}
for all $\rho\in\sscript (H)$, $x\in\Omega _\iscript$, $y\in\Omega _\jscript$. Notice that $\Omega _{\iscript\circ\jscript}=\Omega _\iscript\times\Omega _\jscript$. We call the 2-marginal
\begin{equation*}
(\jscript\mid\iscript )_y(\rho )=(\iscript\circ\jscript )_y^2(\rho )=\sum _x(\iscript\circ\jscript )_{xy}(\rho )=\sum _x\jscript _y\paren{\iscript _x(\rho )}
   =\jscript _y\paren{\,\iscriptbar (\rho )}
\end{equation*}
the \textit{instrument} $\jscript$ \textit{given} (or \textit{conditioned by} or \textit{in the context of}) $\iscript$ and we call the 1-marginal
\begin{equation*}
(\iscript\rmt\jscript )_x(\rho )=(\iscript\circ\jscript )_x^1(\rho )=\sum _y(\iscript\circ\jscript )_{xy}(\rho )=\sum _y\jscript _y\paren{\iscript _x(\rho )}
   =\jscriptbar\paren{\iscript _x(\rho )}
\end{equation*}
the \textit{instrument} $\iscript$ \textit{then} $\jscript$ \cite{gud120,gud22}. If $\kscript\in\instr (H,H_1\otimes H_2)$ we have the \textit{reduced instruments}
$\kscript _1\in\instr(H,H_1)$, $\kscript _2\in\instr (H,H_2)$ given by the partial traces $\kscript _{1x}(\rho )=\trace _{H_2}\sqbrac{\kscript _x(\rho )}$,
$\kscript _{2x}(\rho )=\trace _{H_1}\sqbrac{\kscript _x(\rho )}$. Notice that $\kscript _1,\kscript _2$ have the same $\rho$-distributions for all $\rho\in\sscript (H)$.

If $\iscript _i\in\instr (H,H_1)$, $i=1,2,\ldots ,n$, with the same outcome space $\Omega$ and $\lambda _i\in\sqbrac{0,1}$ with
$\sum\limits _{i=1}^n\lambda _i=1$, then $\iscript =\sum\limits _{i=1}^n\lambda _i\iscript _i$ given by $\iscript _x=\sum\limits _{i=1}^n\lambda _i\iscript _{1x}$,
$x\in\Omega$, is called a \textit{convex combination} of the $\iscript _i$ \cite{gud220}. We have that
\begin{equation*}
\Phi _\rho ^{\sum\lambda _i\iscript _i}(x)=\trace\sqbrac{\sum _{i=1}^n\lambda _i\iscript _{ix}(\rho )}=\sum _{i=1}^n\lambda _i\trace\sqbrac{\iscript _{ix}(\rho )}
   =\sum _{i=1}^n\lambda _i\Phi _x^{\iscript _i}(\rho )
\end{equation*}
for all $\rho\in\sscript (H)$. Thus, the distribution of a convex combination is the convex combination of the distributions. Convex combinations are an important way of combining instruments. We now consider another important way. If $\iscript\in\instr (H,H_1)$ and $\lambda _{xz}\in\sqbrac{0,1}$ with $\sum\limits _z\lambda _{xz}=1$ for all $x\in\Omega _\iscript$, then the instrument $\pscript\in\instr (H,H_1)$ given by $\pscript _z(\rho )=\sum\limits _x\lambda _{xz}\iscript _x(\rho )$ is called a
\textit{post-processing} of $\iscript$ \cite{dapt22,hz12}. Two instruments $\iscript\in\instr (H,H_1)$ and $\jscript\in\instr (H_,H_2)$ \textit{coexist}
(\textit{are compatible}) \cite{ls22}, denoted by $\iscript\rmco\jscript$, if there exists a \textit{joint bi-instrument} $\kscript\in\instr (H,H_1\otimes H_2)$ with
$\Omega _\kscript =\Omega _\iscript\times\Omega _\jscript$ such that for all $x\in\Omega _\iscript$, $y\in\Omega _\jscript$, $\rho\in\sscript (H)$ we have
\begin{align*}
\kscript _{1x}^1(\rho )&=\sum _{y\in\Omega _\jscript}\trace _{H_2}\sqbrac{\kscript _{xy}(\rho )}=\iscript _x(\rho )\\
\kscript _{2y}^2(\rho )&=\sum _{x\in\Omega _\iscript}\trace _{H_1}\sqbrac{\kscript _{xy}(\rho )}=\jscript _y(\rho )
\end{align*}
Thus, two coexisting instruments can be constructed from the same bi-instrument so they are simultaneously measurable. A complete discussion of this concept is found in \cite{dapt22,bkmpt22,mt122,mt222}.
\begin{lem}    
\label{lem21}
If $\iscript\rmco\jscript$ and $\pscript$ is a post-processing of $\iscript$, then $\pscript\rmco\jscript$.
\end{lem}
\begin{proof}
Suppose $\iscript\in\instr (H,H_1)$ and $\jscript\in\instr (H,H_2)$ and let $\kscript\in\instr (H,H_1\otimes H_2)$ be a joint bi-instrument for $\iscript ,\jscript$. If
$\pscript _z=\sum\limits _x\lambda _{xz}\iscript _x$ is a post-processing of $\iscript$, define the bi-instrument $\lscript _{zy}=\sum\limits _x\lambda _{xz}\kscript _{xy}$.
We then obtain
\begin{align*}
\lscript _{1z}^1(\rho )&=\sum _y\trace _{H_2}\sqbrac{\lscript _{zy}(\rho )}=\sum _{x,y}\lambda _{xz}\trace _{H_2}\sqbrac{\kscript _{xy}(\rho )}\\
   &=\sum _x\lambda _{xz}\kscript _{1z}^1(\rho )=\sum _x\lambda _{xz}\iscript _x(\rho )=\pscript _z(\rho )
\intertext{and}
\lscript _{2y}^2(\rho )&=\sum _z\trace _{H_1}\sqbrac{\lscript _{zy}(\rho )}=\sum _{x,z}\lambda _{xz}\trace _{H_1}\sqbrac{\kscript _{xy}(\rho )}\\
   &=\sum _x\trace _{H_1}\sqbrac{\kscript _{xy}(\rho )}=\kscript _{2y}^2(\rho )=\jscript _y(\rho )
\end{align*}
Hence, $\lscript$ is a joint bi-instrument for $\pscript$ and $\jscript$ so $\pscript\rmco\jscript$.
\end{proof}

If $A,B\in\lscript (H)$ satisfy $\elbows{\phi ,A\phi}\le\elbows{\phi ,B\phi}$ for all $\phi\in H$ we write $A\le B$ and if $0\le a\le I$ we call $a$ an \textit{effect}. 
An effect corresponds to a true-false (\ityes-\itno) experiment and $0,I$ are the effects that are always false or always true, respectively. We denote the set of effects on $H$ by $\escript (H)$. If $\rho\in\sscript (H)$, $a\in\escript (H)$, the $\rho$-\textit{probability} of $a$ is $\trace (\rho a)$. Thus, $\trace (\rho a)$ is the probability that $a$ is true (has result \ityes ) when the system is in state $\rho$. If $a$ is true, then its \textit{complement} $a'=I-a\in\escript (H)$ is false. An \textit{observable} is a finite set of effects $A=\brac{A_x\colon x\in\Omega _A}$, $A_x\in\escript (H)$, that satisfies $\sum\limits _{x\in\Omega _A}A_x=I$. We call $\Omega _A$ the
\textit{outcome space} for $A$ and denote the set of observables on $H$ by $\ob (H)$. An observable is also called a \textit{positive operator-valued measure} (POVM)
\cite{dl70,hz12,nc00}. If $\rho\in\sscript (H)$ the $\rho$-\textit{probability distribution} of $A\in\ob (H)$ is gven by $\Phi _\rho ^A(x)=\trace (\rho A_x)$, $x\in\Omega _A$. The observable \textit{measured} by $\iscript\in\instr (H,H_1)$ is the unique $\iscripthat\in\ob (H)$ satisfying $\trace (\rho\iscripthat _x)=\trace\sqbrac{\iscript _x(\rho )}$ for all $\rho\in\sscript (H)$. Since $\trace\sqbrac{\iscript _x(\rho )}=\trace{\rho\iscript _x^*(I_{H_1})}$ we see that $\iscripthat _x=\iscript _x^*(I_{H_1})$ for all
$x\in\Omega _\iscript =\Omega _{\iscripthat}$. We also have the distribution
\begin{equation*}
\Phi _\rho ^{\iscripthat}(x)=\trace (\rho\iscripthat _x)=\trace\sqbrac{\iscript _x(\rho )}=\Phi _\rho ^\iscript (x)
\end{equation*}
for all $x\in\Omega _\iscript =\Omega _{\iscripthat}$. Although an instrument measures a unique observable, as we shall see, an observable is measure by many instruments.

Let $A,B\in\ob (H)$ and suppose $\iscript\in\instr (H,H_1)$ with $\iscripthat =A$. We define the $\iscript$-\textit{sequential product of} $A$ \textit{then} $B$ to be the observable $A\sqbrac{\iscript}B\in\ob (H)$ given by
\begin{equation*}
\paren{A\sqbrac{\iscript}B}_y=\sum _x\iscript _x^*(B_y)
\end{equation*}
As with instruments a \textit{bi-observable} is an observable of the form
\begin{equation*}
A=\brac{A_{xy}\colon (x,y)\in\Omega _1\times\Omega _2}
\end{equation*}
If $B\in\ob (H)$, $\iscript\in\instr (H,H_1)$, we define $B$ \textit{given} $\iscript$ to be the bi-observable $(B\mid\iscript )_{xy}=\iscript _x^*(B_y)$. We then have
$\paren{A\sqbrac{\iscript}B}_y=\sum _x(B\mid\iscript )_{xy}$. Two observables $A, B\in\ob (H)$ \textit{coexist}, denoted $A\rmco B$, if there exists a \textit{joint}
bi-observable $C\in\ob (H)$ with marginals $C_x^1=\sum\limits _yC_{xy}=A_x$ and $C_y^2=\sum _xC_{xy}=B_y$ \cite{dapt22,bkmpt22,hz12,ls22,mt122,mt222}

\begin{lem}    
\label{lem22}
{\rm{(i)}}\enspace If $\iscript\in\instr (H,H_1)$, $\jscript\in\instr (H_1,H_2)$, then $(\iscript\circ\jscript )^*=\jscript ^*\circ\iscript ^*$.
{\rm{(ii)}}\enspace If $\iscript\in\instr (H,H_1)$, $\jscript\in\instr (H_1,H_2)$ and $\iscript\rmco\jscript$, then $\iscripthat\rmco\jscripthat$.
{\rm{(iii)}}\enspace Let $\iscript\in\instr (H_,H_1)$ be a convex combination $\iscript =\sum\lambda _i\iscript _i$. Then
$\iscriptbar =\sum\limits _i\lambda _i\iscriptbar _i$ and $\paren{\sum\limits _i\lambda _i\iscript _i}^\wedge =\sum\limits _i\lambda _i\iscripthat _i$.
\end{lem}
\begin{proof}
(i)\enspace For all $\rho\in\sscript (H)$, $T\in\lscript (H_2)$ we have
\begin{align*}
\trace\sqbrac{\rho\jscript ^*\circ\iscript ^*(T)}&=\trace\sqbrac{\rho\iscript ^*\paren{\jscript ^*(T)}}=\trace\sqbrac{\iscript (\rho )\jscript ^*(T)}
   =\trace\sqbrac{\jscript\paren{\iscript (\rho )}T}\\
   &=\trace\sqbrac{(\iscript\circ\jscript )(\rho )T}=\trace\sqbrac{\rho (\iscript\circ\jscript )^*(T)}
\end{align*}
and the result follows.
(ii)\enspace Since $\iscript\rmco\jscript$, there exists a bi-instrument $\kscript\in\instr (H,H_1\otimes H_2)$ such that $\kscript _{1x}^1=\iscript _x$,
$\kscript _{2y}^2=\jscript _y$. Define the bi-observable $C_{xy}\in\ob (H)$ by $C_{xy}=\kscripthat _{xy}$. Then for all $\rho\in\sscript (H)$ we obtain
\begin{align*}
\trace\paren{\rho\sum _yC_{xy}}&=\trace\paren{\rho\sum _y\kscripthat _{xy}}=\sum _y\trace\sqbrac{\kscript _{xy}(\rho )}
   =\sum _y\trace\sqbrac{\trace _{H_2}\paren{\kscript _{xy}(\rho )}}\\
   &=\trace\sqbrac{\trace _{H_2}\paren{\sum _y\kscript _{xy}(\rho )}}=\trace\sqbrac{\kscript _{1x}^1(\rho )}=\trace\sqbrac{\iscript _x(\rho )}=\trace (\rho\iscripthat _x)
\end{align*}
Hence, $\sum\limits _yC_{xy}=\iscripthat _x$ and similarly $\sum\limits _xC_{xy}=\jscripthat _y$ so $\iscripthat\rmco\jscripthat$.
(iii)\enspace We have that 
\begin{equation*}
\iscriptbar=\sum _x\iscript _x=\sum _x\sum _i\lambda _i\iscript _{ix}=\sum _i\lambda _i\sum _x\iscript _{ix}=\sum _i\lambda _i\iscriptbar _i
\end{equation*}
Moreover, for all $\rho\in\sscript (H)$ we obtain
\begin{align*}
\trace\sqbrac{\rho\paren{\sum _i\lambda _i\iscript _i}^\wedge}&=\trace\sqbrac{\sum _i\lambda _i\iscript _i(\rho )}=\sum _i\lambda _i\trace\sqbrac{\iscript _i(\rho )}\\
   &=\sum _i\lambda _i(\rho\iscripthat _i)=\trace\paren{\rho\sum _i\lambda _i\iscripthat _i}
\end{align*}
so $\paren{\sum\limits _i\lambda _i\iscript _i}^\wedge =\sum\limits _i\lambda _i\iscripthat _i$.
\end{proof}

For a bi-instrument $\kscript\in\instr (H,H_1\otimes H_2)$ we defined the marginals $\kscript _1^1$ and $\kscript _2^2$. We also have the \textit{mixed marginals}
$\kscript _{1x}^2,\kscript _{2y}^1$ given by 
\begin{align*}
\kscript _{1y}^2(\rho )&=\sum _{x\in\Omega _\iscript}\trace _{H_2}\sqbrac{\kscript _{xy}(\rho )}\\
\kscript _{2x}^1(\rho )&=\sum _{y\in\Omega _\jscript}\trace _{H_1}\sqbrac{\kscript _{xy}(\rho )}
\end{align*}

\begin{exam}{1}  
The simplest example of an instrument is a \textit{trivial instrument} $\jscript\in\instr (H,H_2)$ given by $\jscript _y(\rho )=\beta _y$ for all $\rho\in\sscript (H)$, where
$\beta _y\in\escript (H_2)$ with $\beta =\sum\beta _y\in\sscript (H_2)$. Then $\iscript\rmco\jscript$ for all $\iscript\in\instr (H,H_1)$. Indeed, let
$\kscript\in\instr (H,H_1\otimes H_2)$ be the bi-instrument $\kscript _{xy}(\rho )=\iscript _x(\rho )\otimes\beta _y$, $x\in\Omega _\iscript$. Then for all
$\rho\in\sscript (H)$ we have
\begin{align*}
\kscript _{1x}^1(\rho )&=\sum _y\trace _{H_2}\sqbrac{\kscript _{xy}(\rho )}=\sum _y\trace _{H_2}\sqbrac{\iscript _x(\rho )\otimes\beta _y}=\iscript _x(\rho )\\
\kscript _{2y}^2(\rho )&=\sum _x\trace _{H_1}\sqbrac{\kscript _{xy}(\rho )}=\sum _x\trace _{H_1}\sqbrac{\iscript _x(\rho )\otimes\beta _y}=\beta _y=\jscript _y(\rho )
\end{align*}
Hence, $\kscript$ is a joint instrument for $\iscript$ and $\jscript$ so $\iscript\rmco\jscript$. Notice that the mixed marginals of $\kscript$ become:
\begin{align*}
\kscript _{1y}^2(\rho )&=\sum _x\trace _{H_2}\sqbrac{\kscript _{xy}(\rho )}=\sum _x\trace_{H_2}\sqbrac{\iscript _x(\rho )\otimes\beta _y}
  =\trace _{H_2}\sqbrac{\,\iscriptbar (\rho )\otimes\beta _y}\\
  &=\trace (\beta _y)\iscriptbar (\rho )\\
\kscript _{2x}^1(\rho )&=\sum _y\trace _{H_1}\sqbrac{\kscript _{xy}(\rho )}=\sum _y\trace_{H_1}\sqbrac{\iscript _x(\rho )\otimes\beta _y}
   =\trace\sqbrac{\,\iscript _x(\rho )}\sum _y\beta _y\\
   &=\trace\sqbrac{\iscript _x(\rho )}\beta
\end{align*}
We also have $\jscriptbar (\rho )=\beta$ for all $\rho\in\sscript (H)$ and since
\begin{equation*}
\trace (\rho\jscripthat _y)=\trace\sqbrac{\jscript _y(\rho )}=\trace (\beta _y)=\trace\sqbrac{\rho\trace (\beta _x)I_H}
\end{equation*}
we obtain $\jscripthat _y=\trace (\beta _y)I_H$. We call $\jscripthat _y$ an \textit{identity observable} \cite{gud220}.

Let $\jscript\in\instr (H,H_1)$ be a trivial instrument with $\jscript _x(\rho )=\beta _x$, $\beta _x\in\escript (H_1)$. If $\iscript\in\instr (H_1,H_2)$ is arbitrary, we have the sequential product $\jscript\circ\iscript\in\instr (H_1,H_2)$ given by
\begin{equation*}
(\jscript\circ\iscript )_{xy}(\rho )=\iscript _y\paren{\jscript _x(\rho )}=\iscript _y(\beta _x)
\end{equation*}
We then have $\overline{\jscript\circ\iscript}(\rho )=\iscriptbar (\beta )$ for all $\rho\in\sscript (H)$. Since
\begin{equation*}
\trace\sqbrac{\rho (\jscript\circ\iscript )_{xy}^\wedge}=\trace\sqbrac{(\jscript\circ\iscript )_{xy}(\rho )}=\trace\sqbrac{\iscript _y(\beta _x)}
   =\trace\sqbrac{\rho\trace\paren{\iscript _y(\beta _x)}I_H}
\end{equation*}
we obtain $(\jscript\circ\iscript )_{xy}^\wedge =\trace\sqbrac{\iscript _y(\beta _x)}I_H$ which is an identity bi-observable. The conditional instrument
$(\iscript\mid\jscript )\in\instr (H_1,H_2)$ becomes
\begin{equation*}
(\iscript\mid\jscript )_y(\rho )=\iscript _y\paren{\,\jscriptbar (\rho )}=\iscript _y(\beta )
\end{equation*}
for all $\rho\in\sscript (H)$. If $\iscript\in\instr (H_0,H)$ is arbitrary, we have the sequential product $\iscript\circ\jscript\in\instr (H_0,H_1)$ given by
\begin{equation*}
(\iscript\circ\jscript )_{xy}(\rho )=\jscript _y\paren{\iscript _x(\rho )}=\trace\sqbrac{\iscript _x(\rho )}\beta _y
\end{equation*}
We then have $\overline{\iscript\circ\jscript}(\rho )=\beta$ for all $\rho\in\sscript (H_0)$. Since
\begin{align*}
\trace\sqbrac{\rho (\iscript\circ\jscript )_{xy}^\wedge}&=\trace\sqbrac{(\iscript\circ\jscript )_{xy}(\rho )}=\trace\sqbrac{\iscript _x(\rho )\trace (\beta _y)}\\
   &=\trace (\rho\iscripthat _x)\trace (\beta _y)=\trace\sqbrac{\rho\trace (\beta _y)\iscripthat _x}
\end{align*}
we obtain $(\iscript\circ\jscript )_{xy}^\wedge =\trace (\beta _y)\iscripthat _x$. The conditional instrument $(\jscript\mid\iscript )\in\instr (H_0,H_1)$ becomes
\begin{equation*}
(\jscript\mid\iscript )_y(\rho )=\jscript _y\paren{\,\iscriptbar (\rho )}=\beta _y=\jscript _y(\rho )
\end{equation*}
so $(\jscript\mid\iscript )=\jscript$.\hskip 21pc\qedsymbol
\end{exam}

If $A\in\ob (H_1)$, $B\in\ob (H_2)$, define the \textit{tensor product bi-observable} $A\otimes B\in\ob (H_1\otimes H_2)$ by $(A\otimes B)_{xy}=A_x\otimes B_y$
\cite{gud220}. We then have $(A\otimes B)_x^1=A_x\otimes I_{H_2}$, $(A\otimes B)_y^2=I_{H_1}\otimes B_y$ and the identity observables
$(A\otimes B)_{2x}^1=\trace (A_x)I_{H_2}$, $(A\otimes B)_{1y}^2=\trace (B_y)I_{H_1}$. Now $A\otimes B$ is a joint bi-observable for $A,B$ in the sense that
$\tfrac{1}{n_2}(A\otimes B)_{1x}^1=A_x$ and $\tfrac{1}{n_1}(A\otimes B)_{2y}^2=B_y$ where $n_2=\dim H_2$, $n_1=\dim H_1$.

If $\iscript\in\oscript (H_1,H_3)$, $\jscript\in\oscript (H_2,H_4)$, define the \textit{tensor product} $\kscript =\iscript\otimes\jscript$ to be the operation
$\kscript\in (H_1\otimes H_2,H_3\otimes H_4)$ that satisfies
\begin{equation*}
\kscript (C\otimes D)=\iscript (C)\otimes\jscript (D)
\end{equation*}
for all $C\in\lscript (H_1)$, $D\in\lscript (H_2)$. To show that $\kscript$ exists, suppose $\iscript$ and $\jscript$ have Kraus decompositions
$\iscript (C)=\sum\limits _iK_iCK_i^*$, $\jscript (D)=\sum\limits _jJ_jDJ_j^*$ where $\sum\limits _iK_i^*K_i\le I_{H_1}$, $\sum\limits _jJ_j^*J_j\le I_{H_2}$. Then for $E\in\lscript (H_1\otimes H_2)$ we define
\begin{equation*}
\kscript (E)=\sum _{i,j}K_i\otimes J_jEK_i^*\otimes J_j^*
\end{equation*}
Then
\begin{equation*}
\sum _{i,j}(K_i^*\otimes J_j^*)(K_i\otimes J_j)=\sum _{i,j}(K_i^*K_i\otimes J_j^*J_j)=\sum _iK_i^*K_i\otimes\sum _jJ_j^*J_j\le I_{H_1}\otimes I_{H_2}
\end{equation*}
and $\kscript\in\oscript (H_1\otimes H_2,H_3\otimes H_4)$ satisfies
\begin{align*}
K(C\otimes D)&=\sum _{i,j}K_i\otimes J_jC\otimes DK_i^*\otimes J_j^*=\sum _{i,j}(K_iCK_i^*)\otimes (J_jDJ_j^*)\\
   &=\sum _iK_iCK_i^*\otimes\sum _jJ_jDJ_j^*=\iscript (C)\otimes\jscript (D)
\end{align*}
for all $C\in\lscript (H_1)$, $D\in\lscript (H_2)$.

If $\iscript\in\instr (H_1,H_3)$, $J\in\instr (H_2,H_4)$ define the \textit{tensor product} $\kscript =\iscript\otimes\jscript$ to be the bi-instrument
$\kscript\in\instr (H_1\otimes H_2,H_3\otimes H_4)$ defined by $\kscript _{xy}(\rho )=\iscript _x\otimes\jscript _y (\rho )$ for all $\rho\in\sscript (H_1\otimes H_2)$. We have seen that $\kscript _{xy}\in\oscript (H_1\otimes H_2,H_3\otimes H_4)$ and $\kscriptbar$ is a channel because $\kscriptbar=\iscriptbar\otimes\jscriptbar$ and
$\iscriptbar ,\jscriptbar$ are channels. The next result shows that $\iscript\otimes\jscript$ is a type of joint instrument for $\iscript ,\jscript$.

\begin{thm}    
\label{thm23}
Let $\iscript\in\instr (H_1,H_3)$, $\jscript =\instr (H_2,H_4)$ and let $\kscript =\iscript\otimes\jscript$.
{\rm{(i)}}\enspace $\kscripthat _{xy}=\iscripthat _x\otimes\jscripthat _y$.
{\rm{(ii)}}\enspace For all $\rho\in\sscript (H_1\otimes H_2)$ we have $\kscript _{1x}^1(\rho )=\iscript _x\sqbrac{\trace _{H_2}(\rho )}$,
$\kscript _{2y}^2(\rho )=\jscript _y\sqbrac{\trace _{H_1}(\rho )}$.
{\rm{(iii)}}\enspace If $n_1=\dim H_1$, $n_2=\dim H_2$, $\rho _1\in\sscript (H_1)$, $\rho _2=\sscript (H_2)$ we have
\begin{align*}
\tfrac{1}{n_2}\kscript _{1x}^1(\rho _1\otimes I_{H_2})&=\iscript _x(\rho _1)\\
\tfrac{1}{n_1}\kscript _{2y}^2(I_{H_1}\otimes\rho _2)&=\jscript _y(\rho _2)
\end{align*}
\end{thm}
\begin{proof}
(i)\enspace For all $\rho=\rho_1\otimes\rho _2\in\lscript (H_1\otimes H_2)$ we have
\begin{align*}
\trace (\rho\kscripthat _{xy})&=\trace\sqbrac{\kscript _{xy}(\rho )}=\trace\sqbrac{\iscript _x\otimes\jscript _y(\rho _1\otimes\rho _2)}
   =\trace\sqbrac{\iscript _x(\rho _1)\otimes\jscript _y(\rho _2)}\\
   &=\trace\sqbrac{\iscript _x(\rho _1)}\trace\sqbrac{\jscript _y(\rho _2)}=\trace (\rho _1\iscripthat _x)\trace (\rho _2\jscripthat _y)
   =\trace (\rho _1\iscripthat _x\otimes\rho _2\jscripthat _y)\\
   &=\trace (\rho _1\otimes\rho _2\iscripthat _x\otimes\jscripthat _y)=\trace (\rho\iscripthat _x\otimes\jscripthat _y)
\end{align*}
Since any $A\in\lscript (H_1\otimes H_2)$ has the form $A=\sum\limits _{i,j}B_i\otimes C_j$, $B_i\in\lscript (H_1)$, $C_j\in\lscript (H_2)$, the result holds for $\rho =A$. Hence, $\kscripthat _{xy}=\iscripthat _x\otimes\jscripthat _y$.
(ii)\enspace For all $\rho =\rho _1\otimes\rho _2\in\lscript (H_1\otimes H_2)$ we have 
\begin{align*}
\kscript _{1x}^1(\rho )&=\trace _{H_4}\sqbrac{\sum _y\kscript _{xy}(\rho )}=\trace _{H_4}\sqbrac{\sum _y\iscript _x\otimes\jscript _y(\rho _1\otimes\rho _2)}\\
   &=\trace _{H_4}\sqbrac{\sum _y\iscript _x(\rho _1)\otimes\jscript _y(\rho _2)}=\sum _y\trace _{H_4}\sqbrac{\iscript _x(\rho _1)\otimes\jscript _y(\rho _2)}\\
   &=\trace _{H_4}\sqbrac{\iscript _x(\rho _1)\otimes\jscripthat (\rho _2)}=\iscript _x(\rho _1)=\iscript _x\sqbrac{\trace _{H_2}(\rho )}
\end{align*}
As in (i) the result follows for all $\rho\in\sscript (H_1\otimes H_2)$.
(iii)\enspace Applying (i) we obtain
\begin{equation*}
\kscript _{1x}^1(\rho _1\otimes I_{H_2})=\iscript _x\sqbrac{\trace _{H_2}(\rho _1\otimes I_{H_2})}=\iscript _x\sqbrac{\trace (I_{H_2})\rho _1}=n_2\iscript _x(\rho _1)
\end{equation*}
Hence, $\tfrac{1}{n_2}\kscript _{1x}^1(\rho _1\otimes I_{H_2})=\iscript (\rho _1)$. Similarly, $\tfrac{1}{n_1}\kscript _{2y}^2(I_{H_1}\otimes\rho _2)=\jscript _y(\rho _2)$.
\end{proof}

\section{Examples of Instruments}  
Two important instruments are the Holevo and Kraus instruments. These instruments are useful for illustrating the definitions and concepts presented in Section~2. If $A\in\ob (H)$ and $\alpha =\brac{\alpha _x\colon x\in\Omega _A}\subseteq\sscript (H_1)$, the corresponding \textit{Holevo instrument}
$\hscript ^{(A,\alpha )}\in\instr (H,H_1)$ has the form $\hscript _x^{(A,\alpha )}(\rho )=\trace (\rho A_x)\alpha _x$ for all $\rho\in\sscript (H)$ \cite{gud120,hol82}. Notice
that $\hscript ^{(A,\alpha )}$ is indeed an instrument because
\begin{equation*}
\sum _x\trace (\rho A_x)=\trace\paren{\rho\sum _xA_x}=\trace (\rho )=1
\end{equation*}
for every $\rho\in\sscript (H)$ so $\sum _x\trace (\rho A_x)\alpha _x$ is a convex combination of states which is a state. Since
\begin{align*}
\trace\sqbrac{\rho\hscript _x^{(A,\alpha )*}(a)}&=\trace\sqbrac{\hscript _x^{(A,\alpha )}(\rho )a}=\trace\sqbrac{\trace (\rho A_x)\alpha _xa}\\
   &=\trace\sqbrac{\rho\trace (\alpha _xa)A_x}
\end{align*}
we have that $\hscript _x^{(A,\alpha )*}(a)=\trace (\alpha _xa)A_x$ for all $a\in\escript (H_1)$. We conclude thhat 
\begin{equation*}
(\hscript _x^{(A,\alpha )})^\wedge =\hscript _x^{(A,\alpha )*}(I_{H_1})=A_x
\end{equation*}
so $\hscript ^{(A,\alpha )\wedge}=A$. We also have $\hscriptbar ^{(A,\alpha )}(\rho )=\sum\limits _x\trace (\rho A_x)\alpha _x$ which, as we showed previously is a state.

If $\hscript ^{(A,\alpha )}\in\instr (H,H_1)$ and $\hscript ^{(B,\beta )}\in\instr (H_1,H_2)$, then their sequential product becomes
\begin{align*}
\sqbrac{\hscript ^{(A,\alpha )}\circ\hscript ^{(B,\beta )}}_{xy}(\rho )&=\hscript _y^{(B,\beta )}\sqbrac{\hscript _x^{(A,\alpha )}(\rho )}
   =\hscript _y^{(B,\beta )}\sqbrac{\trace(\rho A_x)\alpha _x}\\
   &=\trace (\rho A_x)\hscript _y^{(B,\beta )}(\alpha _x)=\trace (\rho A_x)\trace (\alpha _xB_y)\beta _y\\
   &=\trace\sqbrac{\rho\trace (\alpha _xB_y)A_x}\beta _y=\hscript _{xy}^{(C_y\beta )}(\rho )
\end{align*}
We conclude that $\hscript ^{(A,\alpha )}\circ\hscript ^{(B,\beta )}=\hscript ^{(C,\beta )}$ where $C\in\ob (H)$ is the bi-observable given by
$C_{xy}=\trace (\alpha _xB_y)A_x$. The conditioned instrument $(\hscript ^{(B,\beta )}\mid\hscript ^{(A,\alpha )}\in\instr (H,H_2)$ becomes
\begin{align*}
(\hscript ^{(B,\beta )}\mid\hscript ^{(A,\alpha )})_y(\rho )&=\hscript _y^{(B,\beta )}\sqbrac{\overline{\hscript ^{(A,\alpha )}}\rho}
   =\hscript _y^{(B\beta )}\sqbrac{\sum _x\hscript _x^{(A,\alpha )}(\rho )}\\
   &=\sum _x\hscript _y^{(B,\beta )}\sqbrac{\hscript _x^{(A,\alpha )}(\rho )}=\sum _x\hscript _{xy}^{(C,\beta )}(\rho )=\hscript _y^{(C,\beta )2}(\rho )
\end{align*}
We conclude that $(\hscript ^{(B,\beta )}\mid\hscript ^{(A,\alpha )})$ is the marginal instrument $\hscript ^{(C,\beta )2}$. We also have
\begin{align*}
(\hscript ^{(A,\alpha )}\rmt\hscript ^{(B,\beta )})_x(\rho )&=\overline{\hscript ^{(B,\beta )}}\sqbrac{\hscript _x^{(A,\alpha )}(\rho )}
   =\sum _y\hscript _y^{(B,\beta )}\sqbrac{\hscript _x^{(A,\alpha )}(\rho )}\\
   &=\sum _y\hscript _{xy}^{(C,\beta )}(\rho )=\hscript _x^{(C,\beta )1}(\rho )
\end{align*}
Hence, $(\hscript ^{(A,\alpha )}\rmt\hscript ^{(B,\beta )})$ is the marginal instrument $\hscript ^{(C,\beta )1}$. Notice that $C_x^1=A_x$ so $C^1=A$ and
\begin{equation*}
C_y^2=\sum _x\trace (\alpha _xB_y)A_x
\end{equation*}
Since $\sum\limits _x\trace (\alpha _xB_y)=1$ for every $y\in\Omega _B$, $C^2$ is a post-processing of $A$.

Let $A_{xy}\in\ob (H)$ be a bi-observable, $\alpha =\brac{\alpha _{xy}\colon (x,y)\in\Omega _A}\subseteq\sscript (H_1\otimes H_2)$ and define the Holevo
bi-instrument in $\instr (H,H_1\otimes H_2)$ by
\begin{equation*}
\hscript _{xy}^{(A,\alpha )}(\rho )=\trace (\rho A_{xy})\alpha _{xy}
\end{equation*}
The marginals become
\begin{align*}
\hscript _{xy}^{(A,\alpha )1}(\rho )&=\sum _y\hscript _{xy}^{(A,\alpha )}=\sum_y\trace (\rho A_{xy})\alpha _{xy}\\
\hscript _{xy}^{(A,\alpha )2}(\rho )&=\sum _x\hscript _{xy}^{(A,\alpha )}=\sum _x\trace (\rho A_{xy})\alpha _{xy}
\end{align*}
We then have the reduced and mixed marginals 
\begin{align*}
\hscript _{1x}^{(A,\alpha )1}(\rho )&=\sum _y\trace (\rho A_{xy})\trace _{H_2}(\alpha _{xy})\in\instr (H,H_1)\\
\hscript _{2y}^{(A,\alpha )2}(\rho )&=\sum _x\trace (\rho A_{xy})\trace _{H_1}(\alpha _{xy})\in\instr (H,H_2)\\
\hscript _{1y}^{(A,\alpha )2}(\rho )&=\sum _x\trace (\rho A_{xy})\trace _{H_2}(\alpha _{xy})\in\instr (H,H_1)\\
\hscript _{2x}^{(A,\alpha )1}(\rho )&=\sum _y\trace (\rho A_{xy})\trace _{H_1}(\alpha _{xy})\in\instr (H,H_2)
\end{align*}
We say that $\hscript ^{(A,\alpha )}$ is a \textit{product instrument} if $\alpha _{xy}=\beta _x\otimes\gamma _y$, $\beta _x\in\sscript (H_1)$,
$\gamma _y\in\sscript (H_2)$ and in this case we have
\begin{align*}
\hscript _{1x}^{(A,\alpha )1}(\rho )&=\sum _y\trace (\rho A_{xy})\beta _x\\
\hscript _{2y}^{(A,\alpha )2}(\rho )&=\sum _x\trace (\rho A_{xy})\gamma _y
\end{align*}
Notice that $\hscript _{1x}^{(A,\alpha )1}=\hscript _x^{(B,\beta )}$ where $B_x=\sum\limits _yA_{xy}=A_x^1$ and
$\hscript _{2y}^{(A,\alpha )2}=\hscript _y^{(C,\gamma )}$ where $C_y=\sum\limits _xA_{xy}=A_y^2$.

Suppose $\hscript ^{(A,\alpha )}\in\instr (H,H_1)$, $\hscript ^{(B,\beta )}\in\instr (H,H_2)$ and $\hscript ^{(A,\alpha )}$ so $\hscript ^{(B,\beta )}$. If their joint instrument is $\hscript ^{(C,\gamma )}\in\instr (H,H_1\otimes H_2)$ then for all $\rho\in\sscript (H)$ we have
\begin{align*}
\trace (\rho A_x)\alpha _x&=\hscript _x^{(A,\alpha )}(\rho )=\hscript _{1x}^{(C,\gamma )1}=\sum _y\trace (\rho C_{xy})\trace _{H_2}(\gamma _{xy})\\
\trace (\rho B_y)\beta _y&=\hscript _y^{(B,\beta )}(\rho )=\hscript _{2y}^{(C,\gamma )2}=\sum _x\trace (\rho C_{xy})\trace _{H_1}(\gamma _{xy})
\end{align*}
If $C$ is a product instrument with $\gamma _{xy}=\varepsilon _x\otimes\delta _y$ we obtain
\begin{align*}
\trace (\rho A_x)\alpha _x&=\sum _y\trace (\rho C_{xy})\varepsilon _x=\trace\paren{\rho\sum _yC_{xy}}\varepsilon _x=\trace (\rho C_x^1)\varepsilon _x\\
\trace (\rho B_y)\beta _y&=\sum _x\trace (\rho C_{xy})\delta _y=\trace\paren{\rho\sum _xC_{xy}}\delta _y=\trace (\rho C_y^2)\delta _y
\end{align*}
It follows that $\varepsilon _x=\alpha _x$, $A_x=C_x^1$ and $\beta _y=\delta _y$, $B_y=C_y^2$. Moreover, $\gamma _{xy}=\alpha _x\otimes\beta _y$.

A \textit{Kraus instrument} $\kscript\in\instr (H,H_1)$ has the form $\kscript _x(\rho )=K_x\rho K_x^*$ where $K_x\colon\lscript (H)\to\lscript (H_1)$ are linear operators satisfying $\sum\limits _xK_x^*K_x=I_H$ \cite{kra83}. We call $K_x$ the \textit{Kraus operators} for $\kscript$. Notice that $0\le K_x^*K_x\le I_H$ so
$K_x^*K_x\in\escript (H)$ for all $x\in\Omega _\kscript$. Since
\begin{equation*}
\trace\sqbrac{\kscript _x(\rho )a}=\trace (K_x\rho K_x^*a)=\trace (\rho K_x^*aK_x)
\end{equation*}
for every $a\in\lscript (H_1)$ we have $\kscript _x^*(a)=K_x^*aK_x$. It follows that the measured observable $\kscripthat\in\ob (H)$ is
\begin{equation*}
\kscripthat _x=\kscript _x^*(I_{H_1})=K_x^*K_x
\end{equation*}
for all $x\in\Omega _\kscript$. Let $\kscript\in\instr(H,H_1)$, $\jscript\in\instr (H_1,H_2)$ be Kraus instruments with operators $K_x$, $J_y$, respectively.
Then $\kscript\circ\jscript\in\instr (H,H_2)$ is the bi-instrument given by
\begin{equation*}
(\kscript\circ\jscript )_{xy}(\rho )=\jscript _y\sqbrac{\kscript _x(\rho )}=J_y(K_x\rho K_x^*)J_y^*=J_yK_x\rho (J_yK_x)^*=\lscript _{xy}(\rho )
\end{equation*}
where $\lscript _{xy}$ is the Kraus bi-instrument with Kraus operators $L_{xy}=J_yK_x$. It follows that $(\jscript\mid\kscript )\in\instr (H,H_2)$ is given by
\begin{align*}
(\jscript\mid\kscript )_y(\rho )&=\jscript _y\paren{\,\kscriptbar (\rho )}=\jscript _y\paren{\sum _xK_x\rho K_x^*}=\sum _x\sqbrac{\jscript _y(K_x\rho K_x^*)}\\
   &=\sum _x(J_yK_x\rho K_x^*J_y^*)=\sum _x\lscript _{xy}(\rho )=\lscript _y^2(\rho )
\end{align*}
We also have
\begin{align*}
(\kscript\rmt\jscript )_x(\rho )&=\jscriptbar\sqbrac{\kscript _x(\rho )}=\sum _y\jscript _y(K_x\rho K_x^*)=\sum _yJ_yK_x\rho K_x^*J_y^*\\
   &=\sum _y\lscript _{xy}(\rho )=\lscript _x^1(\rho )
\end{align*}

Let $\hscript ^{(A,\alpha )}\in\instr (H_1,H_2)$ be Holevo and $\kscript\in\instr (H,H_1)$ be an arbitrary instrument. We then have the bi-instrument $\kscript\circ\hscript ^{(A,\alpha )}\in\instr (H,H_2)$ as follows
\begin{align*}
(\kscript\circ\hscript ^{(A,\alpha )})_{xy}(\rho )&=\hscript _y^{(A,\alpha )}(\kscript _x(\rho ))=\trace\sqbrac{\kscript _x(\rho )A_y}\alpha _y\\
   &=\trace\sqbrac{\rho\kscript _x^*(A_y)}\alpha _y=\hscript _{xy}^{(B,\alpha )}(\rho )
\end{align*}
where $B\in\ob (H)$ is the bi-observable given by $B_{xy}=\kscript _x^*(A_y)$. We conclude that
\begin{equation*}
(\kscript\circ\hscript ^{(A,\alpha )})_{xy}^\wedge = B_{xy}=\kscript _x^*(A_y)
\end{equation*}
We also have
\begin{align*}
(\hscript ^{(A,\alpha )}\mid\kscript )_y(\rho )&=\hscript _y^{(A,\alpha )}(\,\kscriptbar (\rho ))=\hscript _y^{(A,\alpha )}\sqbrac{\sum _x\kscript _x(\rho )}\\
   &=\trace\sqbrac{\rho\sum _x\kscript _x^*(A_y)}\alpha _y=\trace (\rho B_y^2)\alpha _y=\hscript _y^{(B^2,\alpha )}(\rho )
\end{align*}
Hence, $(\hscript ^{(A,\alpha )}\mid\kscript )=\hscript ^{(B^2,\alpha )}$ which is Holevo. Moreover,
\begin{align*}
(\kscript\rmt\hscript ^{(A,\alpha )})_x(\rho )&=\overline{\hscript ^{(A,\alpha )}}\sqbrac{\kscript _x(\rho )}=\sum _y\trace\sqbrac{\kscript _x(\rho )A_y}\alpha _y\\
   &=\sum _y\trace\sqbrac{\rho\kscript _x^*(A_y)}\alpha _y=\sum _y\trace\sqbrac{\rho B_{xy}}\alpha _y\\
   &=\sum _y\hscript _{xy}^{(B,\alpha )}(\rho )=\hscript _x^{(B,\alpha )1}(\rho )
\end{align*}
Therefore, $(\kscript\rmt H^{(A,\alpha )})=\hscript ^{(B,\alpha )1}$ which is a marginal of a Holevo bi-instrument. We conclude that the sequential product of an arbitrary instrument then a Holevo instrument is Holevo and a Holevo instrument conditioned by an arbitrary instrument is Holevo. In particular, if $\kscript$ is Kraus with operators $K_x$, then $\kscript\circ\hscript ^{(A,\alpha )}=\hscript ^{(B,\alpha )}$ where $B_{xy}=K_x^*A_yK_x$.

In the other order, let $\hscript ^{(A,\alpha )}\in\instr (H,H_1)$ and $\kscript\in\instr (H_1,H_2)$ be arbitrary. Then $\hscript ^{(A,\alpha )}\circ\kscript\in\instr (H,H_2)$ is the bi-instrument given by
\begin{equation*}
(\hscript ^{(A,\alpha )}\circ\kscript )_{xy}(\rho )=\kscript _y\sqbrac{\hscript _x^{(A,\alpha )}(\rho )}=\kscript _y\sqbrac{\trace (\rho A_x)\alpha _x}
   =\trace (\rho A_x)\kscript _y(\alpha _x)
\end{equation*}
If $\kscript _y(\alpha _x)\ne 0$, let $\beta _{xy}\in\sscript (H_2)$ be defined by $\beta _{xy}=\kscript _y(\alpha _x)/\trace\sqbrac{\kscript _y(\alpha _y)}$ and define the bi-observable $B_{xy}=\trace\sqbrac{\kscript _y(\alpha _x)}A_x$. We then obtain
\begin{equation*}
(\hscript ^{(A,\alpha )}\circ\kscript )_{xy}(\rho )=\trace\sqbrac{\kscript _y(\alpha _x)}\trace (\rho A_x)\beta _{xy}=\trace (\rho B_{xy})\beta _{xy}
   =\hscript _{xy}^{(B,\beta )}(\rho )
\end{equation*}
which is a Holevo bi-instrument. Hence,
\begin{equation*}
(\hscript ^{(A,\alpha )}\circ\kscript )_{xy}^\wedge =B_{xy}=\trace\sqbrac{\kscript _y(\alpha _x)}A_x
\end{equation*}
We also have
\begin{align*}
(\kscript\mid\hscript ^{(A,\alpha )})_y(\rho )&=\kscript _y(\overline{\hscript ^{(A,\alpha )}}(\rho ))=\kscript _y\sqbrac{\sum _x\hscript _x^{(A,\alpha )}(\rho )}
   =\sum _x\kscript _y\sqbrac{\trace (\rho A_x)\alpha _x}\\
   &=\sum _x\trace (\rho A_x)\kscript _y(\alpha _x)=\sum _x\trace (\rho A_x)\trace\sqbrac{\kscript _y(\alpha _x)}\beta _{xy}\\
   &=\sum _x\trace (\rho B_{xy})\beta _{xy}=\sum _x\hscript _{xy}^{(B,\beta )}(\rho )=\hscript _y^{(B,\beta )2}(\rho )
\end{align*}
Therefore, $(\kscript\mid\hscript ^{(A,\alpha )})=\hscript ^{(B,\beta )2}$ which is a marginal of a Holevo bi-instrument. Moreover,
\begin{align*}
(\hscript ^{(A,\alpha )}\rmt\kscript )_x(\rho )&=\kscriptbar\sqbrac{\hscript _x^{(A,\alpha )}(\rho )}=\sum _y\kscript _y\sqbrac{\trace (\rho A_x)\alpha _x}
   =\trace (\rho A_x)\sum _y\kscript _y(\alpha _x)\\
   &=\trace (\rho A_x)\sum _y\trace\sqbrac{\kscript _y(\alpha _x)}\beta _{xy}=\sum _y\trace\brac{\rho\trace\sqbrac{\kscript _y(\alpha _x)}A_x}\beta _{xy}\\
   &=\sum _y\trace (\rho B_{xy})\beta _{xy}=\sum _y\hscript _{xy}^{(B,\beta )}(\rho )=\hscript _x^{(B,\beta )1}(\rho )
\end{align*}
Hence, $(\hscript ^{(A,\alpha )}\rmt\kscript )=\hscript ^{(B,\beta )1}$ which is a marginal of a Holevo bi-instrument.

We now give an example of a convex tensor product of two instruments. Let $\iscript\in\instr (H,H_1)$, $\jscript\in\instr (H,H_2)$, $\alpha _x\in\sscript (H_1)$,
$\beta _y\in\sscript (H_2)$, $\lambda _y,\mu _x\in\sqbrac{0,1}$ with $\sum _y\lambda _y+\sum _x\mu _x=1$ and define $\lambda =\sum\limits _y\lambda _y$,
$\mu =\sum\limits _x\mu _x$. Define the bi-instrument $\kscript\in\instr (H,H_1\otimes H_2)$ by
\begin{equation*}
\kscript _{xy}(\rho )=\lambda _y\iscript _x(\rho )\otimes\beta _y+\mu _x\alpha _x\otimes\jscript _y(\rho )
\end{equation*}
Notice that $\kscript$ is indeed an instrument because
\begin{align*}
\trace\sqbrac{\sum _{x,y}\kscript _{xy}(\rho )}&=\sum _{x,y}\trace\sqbrac{\kscript _{xy}(\rho )}
   =\sum _{x,y}\brac{\lambda _y\trace\sqbrac{\iscript _x(\rho )}+\mu _x\trace\sqbrac{\jscript _y(\rho )}}\\
   &=\sum _y\lambda _y\trace\sqbrac{\,\iscriptbar (\rho )}+\sum _x\mu _x\trace\sqbrac{\,\jscriptbar (\rho )}=\sum _y\lambda _y+\sum _x\mu _x=1
\end{align*}
The marginals $\kscript ^1\in\instr (H,H_1\otimes H_2)$, $\kscript ^2\in\instr (H,H_1\otimes H_2)$ are given by
\begin{align*}(
\kscript _x^1(\rho )&=\sum _y\kscript _{xy}(\rho )=\iscript _x(\rho )\otimes\sum _y\lambda _y\beta _y+\mu _x\alpha _x\otimes\jscriptbar (\rho )\\
\kscript _y^2(\rho )&=\sum _x\kscript _{xy}(\rho )=\iscriptbar (\rho )\otimes\lambda _y\beta _y+\sum _x\mu _x\alpha _x\otimes\jscript _y (\rho )
\end{align*}
The reduced instruments $K_1\in\instr (H,H_1)$, $\kscript _2\in\instr (H,H_2)$ become
\begin{align*}
\kscript _{1xy}(\rho )&=\trace _{H_2}\sqbrac{\kscript _{xy}(\rho )}=\lambda _y\iscript _x(\rho )+\mu _x\trace\sqbrac{\jscript _y(\rho )}\alpha _x\\
\kscript _{2xy}(\rho )&=\trace _{H_1}\sqbrac{\kscript _{xy}(\rho )}=\lambda _y\trace\sqbrac{\jscript _x(\rho )}\beta _y+\mu _x\jscript _y(\rho )
\end{align*}
The reduced marginals $K_1^1\in\instr (H,H_1)$, $\kscript _2^2\in\instr (H,H_2)$ $\kscript _1^2\in\instr (H,H_1)$, $\kscript _2^1\in\instr (H,H_2)$ are given by
\begin{align*}
\kscript _{1x}^1(\rho )&=\sum _y\kscript _{1xy}(\rho )=\lambda\iscript _x(\rho )+\mu _x\alpha _x\\
\kscript _{2y}^2(\rho )&=\sum _x\kscript _{2xy}(\rho )=\lambda _y\beta _y+\mu\jscript _y(\rho )\\
\kscript _{1y}^2(\rho )&=\sum _x\kscript _{1xy}(\rho )=\lambda _y\iscriptbar (\rho )+\trace\sqbrac{\jscript _y(\rho )}\sum _x\mu _x\alpha _x\\
\kscript _{2x}^1(\rho )&=\sum _y\kscript _{2xy}(\rho )=\trace\sqbrac{\iscript _x(\rho )}\sum _y\lambda _y\beta _y+\mu _x\jscriptbar (\rho )						\end{align*}
We have that $\kscript _1^1\rmco\kscript _2^2$ and $\kscript _1^2\rmco\kscript _2^1$. The measured observables are gotten as follows:
\begin{align*}
\trace (\rho\kscripthat _{xy})&=\trace\sqbrac{\,\kscript _{xy}(\rho )}=\lambda _y\trace\sqbrac{\iscript _x(\rho )}+\mu _x\trace\sqbrac{\jscript _y(\rho )}\\
   &=\lambda _y\trace (\rho\iscripthat _x)+\mu _x\trace (\rho\jscripthat _y)
\end{align*}
Hence, $\kscripthat _{xy}=\lambda _y\iscripthat _x+\mu _x\jscripthat _y$. Therefore, $\kscripthat _x^1=\lambda\iscripthat _x+\mu _xI_H$ and
$\kscripthat _y^2=\lambda _yI_H+\mu\jscripthat _y$ coexist with joint observable $\kscripthat _{xy}$. We also have
\begin{align*}
\trace (\rho\kscripthat _{1x}^1)&=\trace\sqbrac{\kscript _{1x}^1(\rho )}=\lambda\trace\sqbrac{\iscript _x(\rho )}+\mu _x
   =\lambda\trace (\rho\iscripthat _x)+\mu _x\trace (\rho )\\
   &=\trace\sqbrac{\rho (\lambda\iscripthat _x+\mu _xI_H)}
\end{align*}
Hence, $\kscripthat _{1x}^1=\lambda\iscripthat _x+\mu _xI_H=\kscripthat _x^1$ and similarly $\kscripthat _{2y}^2=\lambda _yI_H+\mu\jscripthat _y=\kscripthat _y^2$. Moreover,
\begin{equation*}
\trace (\rho\kscripthat _{1y}^2)=\trace\sqbrac{\kscript _{1y}^2 (\rho )}=\lambda _y+\mu\trace\sqbrac{\jscript _y(\rho )}
   =\trace\sqbrac{\rho (\mu\jscripthat _y+\lambda _yI_H)}
\end{equation*}
Therefore,
\begin{equation*}
\kscripthat _{1y}^2=\mu\jscripthat _y+\lambda _uI_H=\kscripthat _{2y}^2=\kscripthat _y^2
\end{equation*}
and similarly $\kscripthat _{2x}^1=\kscripthat _{1x}^1=\kscripthat _x^1$.

Let us consider the special case in which $\iscript =\hscript ^{(A,\gamma )}$ and $\jscript =\hscript ^{(B,\delta )}$. We then obtain
\begin{align*}
\kscript _{xy}(\rho )&=\lambda _y\hscript _x^{(A,\gamma )}(\rho )\otimes\beta _y+\mu _x\alpha _x\otimes\hscript _y^{(B,\delta )}(\rho )\\
   &=\lambda _y\trace (\rho A_x)\gamma _x\otimes\beta _y+\mu _x\alpha _x\otimes\trace (\rho B_y)\gamma _y
\end{align*}
In this case, we have
\begin{align*}
\kscript _{1xy}(\rho )&=\lambda _y\trace (\rho A_x)\gamma _x+\mu _x\trace (\rho B_y)\alpha _x\\
\kscript _{2xy}(\rho )&=\lambda _y\trace (\rho A_x)\beta _y+\mu _x\trace (\rho B_y)\delta _y
\end{align*}
We also obtain $\kscripthat _{xy}=\lambda _yA_x+\mu _xB_y$, $\kscripthat _x^1=\lambda A_x+\mu _xI_H$, $\kscripthat _y^2=\lambda _yI_H+\mu B_y$.

\section{Results}  
Our first result show that a convex combination of Holevo instruments with the same base Hilbert space, outcome space and states is Holevo. Moreover, a weaken form of the converse holds.

\begin{thm}    
\label{thm41}
{\rm{(i)}}\enspace Let $\hscript ^{(A_i,\alpha )}$, $i=1,2,\ldots ,n$, be Holevo instruments in $\instr (H,H_1)$ with the same outcome space $\Omega$ and states
$\alpha =\brac{\alpha _x\colon x\in\Omega}$. Then a convex combination $\sum\limits _{i=1}^n\lambda _i\hscript ^{(A_i,\alpha )}$ is Holevo and
\begin{equation*}
\sum _{i=1}^n\lambda _i\hscript ^{(A_i,\alpha )}=\hscript ^{(\sum\lambda _iA_i,\alpha )}
\end{equation*}
{\rm{(ii)}}\enspace If $\hscript ^{(A_i,\alpha _i)}\in\instr (H,H_1)$ with the same outcomes space $\Omega$ and if
\begin{equation*}
\sum _{i=1}^n\lambda _i\hscript ^{(A_i,\alpha _i)}=\hscript ^{(B,\beta )}
\end{equation*}
then $B=\sum\lambda _iA_i$ and
\begin{equation}                
\label{eq41}
\beta _x=\frac{1}{\sum\limits _i\lambda _i\trace (A_{ix})}\sum _i\lambda _i\trace (A_{ix})\alpha _{ix}
\end{equation}
for all $x\in\Omega$.
\end{thm}
\begin{proof}
(i)\enspace For all $x\in\Omega$, we obtain
\begin{align*}
\sum _i\lambda _i\hscript _x^{(A_i,\alpha )}(\rho )&=\sum _i\lambda _i\trace (\rho A_{ix})\alpha _x=\trace\sqbrac{\rho\paren{\sum _i\lambda _iA_i}_x}\alpha _x\\
   &=\hscript _x^{\paren{\sum\lambda _iA_i,\alpha}}(\rho )
\end{align*}
and the result follows.
(ii)\enspace For all $\rho\in\sscript (H)$ and $x\in\Omega$ we have
\begin{equation}                
\label{eq42}
\trace (\rho B_x)\beta _x=\hscript _x^{(B,\beta )}(\rho )=\sum _i\lambda _i\hscript _x^{(A_i,\alpha _i)}(\rho )=\sum _i\lambda _i\trace (\rho A_{ix})\alpha _{ix}
\end{equation}
Taking the trace of \eqref{eq42} gives
\begin{equation*}
\trace (\rho B_x)=\sum _i\lambda _i\trace (\rho A_{ix})=\trace\paren{\rho\sum _i\lambda _iA_{ix}}
\end{equation*}
Hence, $B_x=\sum\limits _i\lambda _iA_{ix}$ for all $x\in\Omega$ and we conclude that $B=\sum\limits _i\lambda _iA_i$. Substituting $B$ into \eqref{eq42} gives
\begin{equation*}
\sum _i\lambda _i\trace (\rho A_{ix})\beta _x=\sum _i\lambda _i\trace (\rho A_{ix})\alpha _{ix}
\end{equation*}
so that
\begin{equation*}
\beta _x=\frac{1}{\sum\limits _i\lambda _i\trace (\rho A_{ix})}\sum _i\lambda _i\trace (\rho A_{ix})\alpha _{ix}
\end{equation*}
for all $x\in\Omega$, $\rho\in\sscript (H)$. Letting $\rho =I/n$ where $n=\dim H$, we conclude that \eqref{eq41} holds.
\end{proof}

We have seen that a convex combination of Holevo instruments $\hscript ^{(A_i,\alpha )}$ is Holevo. We now show that a general convex combination of Holevo instruments $\hscript ^{(A_i,\alpha _i)}$ need not be Holevo.

\begin{exam}{2}  
Let $\hscript ^{(A,\alpha )}$, $\hscript ^{(B,\beta )}\in\instr (\complex ^2)$ be Holevo instruments with the same outcome space $\Omega =\brac{x,y}$ and let
$A_x=B_y=\ket{\phi}\bra{\phi}$ where $\phi\in\complex ^2$ with $\doubleab{\phi}=1$. Also, assume that $\alpha _x\ne\beta _x$ and
\begin{equation*}
\tfrac{1}{2}\hscript ^{(A,\alpha )}+\tfrac{1}{2}\hscript ^{(B,\beta )}=\hscript ^{(C,\gamma )}
\end{equation*}
It follows from Theorem~\ref{thm41}(ii) that $C=\tfrac{1}{2}A+\tfrac{1}{2}B$ so
\begin{equation*}
C_x=\tfrac{1}{2}A_x+\tfrac{1}{2}B_x=\tfrac{1}{2}I=C_y
\end{equation*}
Also from Theorem~\ref{thm41}(ii) we obtain $\gamma _x=\tfrac{1}{2}(\alpha _x+\beta _x)$. Since
\begin{equation*}
\tfrac{1}{2}\trace (\rho A_x)\alpha _x+\tfrac{1}{2}\trace (\rho B_x)\beta _x=\trace (\rho C_x)\gamma _x
\end{equation*}
for all $\rho\in\sscript (\complex ^2)$, letting $\rho=A_x$ we have $\alpha _x=\gamma _x=\tfrac{1}{2}(\alpha _x+\beta _x)$. But then $\alpha _x=\beta _x$ which is a contradiction. Hence, $\tfrac{1}{2}\hscript ^{(A,\alpha )}+\tfrac{1}{2}\hscript ^{(B,\beta )}$ is not Holevo. This also shows that the converse of Theorem~\ref{thm41}(ii) does not hold.\hfill\qedsymbol
\end{exam}

\begin{exam}{3}  
This example shows that a convex combination of Kraus instruments need not be Kraus. Let $\brac{\phi _1,\phi _2}$ be an orthonormal basis for $\complex ^2$, let $K_x,K_y$ be the projection $K_x=\ket{\phi _1}\bra{\phi _1}$, $K_y=\ket{\phi _2}\bra{\phi _2}$ and let $J_x=K_y$, $J_y=K_x$. Define the Kraus instruments
$\kscript ,\jscript\in\instr (\complex ^2)$ with operators $\brac{K_x,K_y}$, $\brac{J_x,J_y}$, respectively. Suppose $\lscript\in\instr (\complex ^2)$ is a Kraus instrument with outcome space $\Omega =\brac{x,y}$, operators $\brac{L_x,L_y}$ so that $L_x^*L_x+L_y^*L_y=I$ and $\lscript =\tfrac{1}{2}\kscript +\tfrac{1}{2}\jscript$. We then obtain
\begin{equation*}
L_x\rho L_x^*=\lscript _x(\rho )=\tfrac{1}{2}\kscript _y (\rho )+\frac{1}{2}\jscript _x(\rho )=\tfrac{1}{2}K_x\rho K_y+\tfrac{1}{2}J_x\rho J_x
\end{equation*}
for all $\rho\in\sscript (\complex ^2)$. Letting $\rho =I/2$ we have
\begin{equation*}
L_xL_x^*=\tfrac{1}{2}K_x+\tfrac{1}{2}J_x=\tfrac{1}{2}I
\end{equation*}
and it follows that $\sqrt{2}L_x$ is a unitary operator $U$. Hence, for all $\rho\in\sscript (\complex ^2)$ we have
\begin{equation*}
U\rho U^*=K_x\rho K_x+J_x\rho J_x
\end{equation*}
Therefore,
\begin{equation*}
K_xU\rho U^*=K_x\rho K_x=U\rho U^*K_x
\end{equation*}
We conclude that $K_x$ commutes with every $\rho\in\sscript (H)$. Hence, $K_x=\lambda _xI$, $\lambda _x\in\sqbrac{0,1}$ which is a contradiction.
\hfill\qedsymbol
\end{exam}

\begin{lem}    
\label{lem42}
If $\jscript\in\instr (H,H_1)$ is a post-processing of a Holevo instrument $\iscript\in\instr (H,H_1)$, then $\jscript$ is Holevo.
\end{lem}
\begin{proof}
Suppose $\iscript =\hscript ^{(A,\alpha )}$ and $\jscript$ is a post-processing of $\iscript$. Then there exist $\lambda _{xy}\in\sqbrac{0,1}$ with
$\sum\limits _y\lambda _{xy}=1$ for all $x\in\Omega _\iscript$ such that
\begin{align*}
\jscript _y(\rho )&=\sum _x\lambda _{xy}\iscript _x(\rho )=\sum _x\lambda _{xy}\hscript _x^{(A,\alpha )}(\rho )=\sum _x\lambda _{xy}\trace (\rho A_x)\alpha _x\\
   &=\trace\paren{\rho\sum _x\lambda _{xy}A_x}\alpha _x=\hscript _y^{\paren{\sum\limits _x\lambda _{xy}A_x,\alpha}}(\rho )
\end{align*}
Hence, $\jscript =\hscript ^{(B,\alpha )}$ is Holevo with $B_y=\sum\limits _x\lambda _{xy}A_x$ a post-processing of $A$.
\end{proof}

We conjecture that Lemma~\ref{lem42} does not hold for Kraus instruments but have not found a counterexample.

\begin{lem}    
\label{lem43}
If $\iscript\in\instr (H,H_1)$, $\jscript\in\instr (H_1,H_2)$, $\kscript\in (H_0,H)$ and $\iscript\rmco\jscript$, then $(\iscript\mid\kscript )\rmco (\jscript\mid\kscript )$. If
$\lscript$ is a joint instrument for $\iscript$ and $\jscript$, then $\mscript =\kscriptbar\circ\lscript$ is a joint instrument for $(\iscript\mid\kscript )$ and
$(\jscript\mid\kscript )$.
\end{lem}
\begin{proof}
Let $\lscript\in\instr (H,H_1\otimes H_2)$ be a joint bi-instrument for $\iscript$, $\jscript$. Define $\mscript\in\instr (H_0,H_1\otimes H_2)$ by
$\mscript _{xy}(\rho )=\lscript _{xy}(\,\kscriptbar (\rho ))$. We then have
\begin{align*}
\mscript _{1x}^1(\rho )&=\lscript _{1x}^1(\,\kscriptbar (\rho ))=\sum _{y\in\Omega _\jscript}\trace _{H_2}\sqbrac{\lscript _{xy}(\,\kscriptbar (\rho ))}
   =\iscript _x(\,\kscriptbar (\rho ))=(\iscript\mid\kscript )_x(\rho )\\
\mscript _{2y}^2(\rho )&=\lscript _{2y}^2(\,\kscriptbar (\rho ))=\sum _{x\in\Omega _\iscript}\trace _{H_1}\sqbrac{\lscript _{xy}(\,\kscriptbar (\rho ))}
   =\jscript _x(\,\kscriptbar (\rho ))=(\jscript\mid\kscript )_y(\rho )
\end{align*}
Hence, $\mscript$ is a joint bi-instrument for $(\iscript\mid\kscript )$ and $(\jscript\mid\kscript )$ so $(\iscript\mid\kscript )\rmco (\jscript\mid\kscript )$. Moreover,
$\mscript =\kscriptbar\circ\lscript$.
\end{proof}

If $\iscript\in\instr (H,H_1)$, then $\iscripthat _x=\iscript _x^*(I_{H_1})\in\ob (H)$ and if $A\in\ob (H_1)$ we define
$(A\mid\iscript )_x=\iscriptbar\,^*(A_x)\in\ob (H)$. Also, if $\jscript\in\instr (H_1,H_2)$ then
\begin{equation*}
(\iscript\circ\jscript )_{xy}(\rho )=\jscript _y(\iscript _x(\rho ))\in\instr (H_1,H_2)
\end{equation*}
and since $(\jscript\mid\iscript )_y(\rho )=\jscript _y(\,\iscriptbar (\rho ))$ we have that $(\jscript\mid\iscript )\in\instr (H_1,H_2)$. Now $\jscripthat\in\ob (H_1)$ so
\begin{equation*}
(\,\jscripthat\mid\iscript )_y=\iscriptbar\,^*(\,\jscripthat _y)\in\ob (H)
\end{equation*}
Also, $(\jscript \mid\iscript )^\wedge\in\ob (H)$ and the next result shows that these two observables coincide.

\begin{lem}    
\label{lem44}
If $\iscript\in\instr (H,H_1)$ and $\jscript\in\instr (H_1,H_2)$, then $(\jscript\mid\iscript )^\wedge=(\,\jscripthat\mid\iscript )$.
\end{lem}
\begin{proof}
For all $y\in\Omega _\jscript$ and $\rho\in\sscript (H)$ we obtain
\begin{align*}
\trace\sqbrac{\rho (\jscript\mid\iscript )_y^\wedge}&=\trace\sqbrac{(\jscript\mid\iscript )_y(\rho )}=\trace\sqbrac{\jscript _y(\,\iscriptbar (\rho ))}
   =\trace\sqbrac{\,\iscriptbar (\rho )\jscripthat _y}\\
   &\trace\sqbrac{\rho\iscriptbar\,^*(\,\jscripthat _y)}=\trace\sqbrac{\rho (\,\jscripthat \mid\iscript )_y}
\end{align*}
Hence, $(\jscript\mid\iscript )^\wedge =(\,\jscripthat\mid\iscript )$.
\end{proof}

\begin{cor}    
\label{cor45}
If $\iscript\in\instr (H,H_1)$, $\jscript\in\instr (H,H_2)$, $\kscript\in\instr (H_0,H)$ and $\iscript\rmco\jscript$, then
$(\,\iscripthat\mid\kscript )\rmco (\,\jscripthat\mid\kscript )$.
\end{cor}
\begin{proof}
By Lemma~\ref{lem43}, $(\iscript\mid\kscript )\rmco (\jscript\mid\kscript )$ so $(\iscript\mid\kscript )^\wedge\rmco (\jscript\mid\kscript )^\wedge$.  By Lemma~\ref{lem44}, $(\,\iscripthat\mid\kscript )=(\iscript\mid\kscript )^\wedge$ and $(\,\jscripthat\mid\kscript )=(\jscript\mid\kscript )^\wedge$ so
$(\,\iscripthat\mid\kscript )\rmco (\,\jscripthat\mid\kscript )$.
\end{proof}

\begin{lem}    
\label{lem46}
Let $A,B\in\ob (H)$ and $\iscript\in\instr (H_1,H)$. If $A\rmco B$, then $(A\mid\iscript )\rmco (B\mid\iscript )$. If $C$ is a joint bi-observable for $A$ and $B$, then
$D_{xy}=\iscriptbar\,^*(C_{xy})$ is a joint bi-observable for $(A\mid\iscript )$ and $(B\mid\iscript )$.
\end{lem}
\begin{proof}
We have that $D,(A\mid\iscript ),(B\mid\iscript )\in\ob (H_1)$ and we obtain
\begin{equation*}
D_x^1=\sum _yD_{xy}=\sum _y\iscriptbar\,^*(C_{xy})=\iscript\,^*\paren{\sum _yC_{xy}}=\iscriptbar\,^*(A_x)=(A\mid\iscript )_x
\end{equation*}
and similarly, $D_y^2=(B\mid\iscript )_y$. Hence, $D$ is a joint bi-observable for $(A\mid\iscript )$ and $(B\mid\iscript )$ so $(A\mid\iscript)\rmco (B\mid\iscript )$.
\end{proof}

\begin{exam}{4}  
The converse of Lemma~\ref{lem46} does not hold. To show this, suppose $A,B\in\ob (H)$ do not coexist. Let $\hscript ^{(C,\alpha )}\in\instr (H_1,H)$ be Holevo with
$C\in\ob (H_1)$, $\brac{\alpha}=\alpha\in\sscript (H)$. Then
\begin{align*}
(A\mid\hscript ^{(C,\alpha )})_x=\hscript ^{(C,\alpha )*}(A_x)&=\sum _z\trace (\alpha A_x)C_z=\trace (\alpha A_x)I_{H_1}\\
(B\mid\hscript ^{(C,\alpha )})_y=\hscript ^{(C,\alpha )*}(B_y)&=\sum _z\trace (\alpha B_y)C_z=\trace (\alpha B_y)I_{H_1}
\end{align*}
Letting $D_{xy}=\trace (\alpha A_x)\trace (\alpha B_y)I_{H_1}\in\ob (H_1)$, we have that $D$ is a joint bi-observable for $(A\mid\hscript ^{(C,\alpha )})$ and
$(B\mid\hscript ^{(C,\alpha )})$. Hence, $(A\mid\hscript ^{(C,\alpha )})\rmco (B\mid\hscript ^{(C,\alpha )})$ but $A,B$ do not coexist.\hfill\qedsymbol
\end{exam}

We say that an observable $A$ is \textit{sharp} if $A_x$ is a projection for all $x\in\Omega _A$ and an instrument $\iscript$ is \textit{sharp} if $\iscripthat$ is sharp
\cite{gud120,hz12,nc00}.

\begin{thm}    
\label{thm47}
Let $\iscript\in\instr (H,H_1)$ and $A\in\ob (H_1)$.
{\rm{(i)}}\enspace $(A\mid\iscript )\rmco\iscripthat$.
{\rm{(ii)}}\enspace If $\iscript$ is sharp, then $(A\mid\iscript )$ commutes with $\iscripthat$.
\end{thm}
\begin{proof}
(i)\enspace Let $B_{xy}$ be the bi-observable on $H$ given by $B_{xy}=\iscript _x^*(A_y)$. Notice that $B_{xy}$ is indeed an observable because
\begin{equation*}
\sum _{x,y}B_{xy}=\sum _{x,y}\iscript _x^*(A_y)=\sum _x\iscript _x^*\paren{\sum _yA_y}=\sum _x\iscript _x^*(I_{H_1})=\sum _x\iscripthat _x=I_H
\end{equation*}
We have that
\begin{align*}
B_x^1&=\sum _yB_{xy}=\iscript _x^*(I_{H_1})=\iscripthat _x\\
B_y^2&=\sum _xB_{xy}=\sum _x\iscript _x^*(A_y)=\iscriptbar\,^*(A_y)=(A\mid\iscript )_y
\end{align*}
so $(A\mid\iscript )\rmco\iscripthat$.
(ii)\enspace If $\iscript$ is sharp, then $\iscripthat$ is sharp and by (i) we have that $\iscripthat\rmco (A\mid\iscript )$. It follows that $\iscripthat _x$ and
$(A\mid\iscript )_y$ are coexisting effects \cite{bgl95,hz12}. Since $\iscripthat _x$ is a projection we conclude that $\iscripthat _x$ and $(A\mid\iscript )_y$ commute for all $x,y$ \cite{bgl95,hz12}.
\end{proof}

\begin{thm}    
\label{thm48}
{\rm{(i)}}\enspace If $\iscript\in\instr (H,H_1)$, $\jscript\in\instr (H_1,H_2)$, then $(\iscript _x\circ\jscript _y)^\wedge =\iscript _x^*(\,\jscripthat _y)$ for all $x,y$.
{\rm{(ii)}}\enspace If $\iscript ,\jscript\in\instr (H)$, then $\iscript\circ\jscript =\jscript\circ\iscript$ implies $\iscript _x^*(\,\jscripthat _y)=\jscript _y^*(\,\iscripthat _x)$ for all $x,y$ which implies $(\iscript\circ\jscript )^\wedge =(\jscript\circ\iscript )^\wedge$.
(iii)\enspace If $\iscript ,\jscript\in\instr (H)$ with $\iscript\circ\jscript =\jscript\circ\iscript$, then $(\,\iscripthat\mid\jscript )=\iscripthat$ and
$(\,\jscripthat\mid\iscript )=\jscripthat$.
\end{thm}
\begin{proof}
(i)\enspace For all $\rho\in\sscript (H)$, we have
\begin{align*}
\trace\sqbrac{\rho (\iscript _x\circ\jscript _y)^\wedge}&=\trace\sqbrac{\iscript _x\circ\jscript _y(\rho )}=\trace\sqbrac{\jscript _y(\iscript _x(\rho ))}\\
   &=\trace\sqbrac{\iscript _x(\rho )\jscripthat _y}=\trace\sqbrac{\rho\iscript _x^*(\,\jscripthat _y)}
\end{align*}
It follows that $(\iscript _x\circ\jscript _y)^\wedge =\iscript _x^*(\jscript _y)$ for $x,y$.
(ii)\enspace If $\iscript\circ\jscript =\jscript\circ\iscript$, then by (i) we obtain
\begin{equation*}
\iscript _x^*(\,\jscripthat _y)=(\iscript _x\circ\jscript _y)^\wedge =(\jscript _y\circ\iscript _x)^\wedge =\jscript _y^*(\,\iscripthat _x)
\end{equation*}
for all $x,y$. Moreover, if $\iscript _x^*(\,\jscripthat _y)=\jscript _y^*(\,\iscripthat _x)$ then by (i) we have
$(\iscript _x\circ\jscript _y)^\wedge =(\jscript _y\circ\iscript _x)^\wedge$.
(iii)\enspace If $\iscript\circ\jscript =\jscript\circ\iscript$, then by (ii) we obtain
\begin{equation*}
\iscripthat _x=\iscript _x^*(I_H)=\sum _y\iscript _x^*(\,\jscripthat _y)=\sum _y\jscript _y^*(\,\iscripthat _x)=(\,\iscripthat\mid\jscript )_x
\end{equation*}
Hence, $\iscripthat =(\,\iscripthat\mid\jscript )$ and similarly, $\jscripthat =(\,\jscripthat\mid\iscript )$.
\end{proof}

\begin{exam}{5}  
Let $\hscript ^{(A,\alpha )},\hscript ^{(B,\beta )}\in\instr (H)$ be Holevo. We have seen in the second paragraph of Section~3 that
\begin{equation*}
\hscript _x^{(A,\alpha )}\circ\hscript _y^{(B,\beta )}(\rho )=\trace (\rho A_x)\trace (\alpha _xB_y)\beta _y
\end{equation*}
and similarly,
\begin{equation*}
\hscript _y^{(B,\beta )}\circ\hscript _x^{(A,\alpha )}(\rho )=\trace (\rho B_y)\trace (\beta _yA_x)\alpha _x
\end{equation*}
Hence, $\hscript _x^{(A,\alpha )}\circ\hscript _y^{(B,\beta )}=\hscript _y^{(B,\beta )}\circ\hscript _x^{(A,\alpha )}$ if and only if
\begin{equation}                
\label{eq43}
\trace (\rho A_x)\trace (\alpha _xB_y)\beta _y=\trace (\rho B_y)\trace (\beta _yA_x)\alpha _x
\end{equation}
for all $\rho\in\sscript (H)$. Taking the trace of \eqref{eq43} gives
\begin{equation}                
\label{eq44}
\trace (\rho A_x)\trace (\alpha _xB_y)=\trace (\rho B_y)\trace (\beta _yA_x)
\end{equation}
for all $\rho\in\sscript (H)$. Applying \eqref{eq44} we have
\begin{equation}                
\label{eq45}
\trace\sqbrac{\rho\trace (\alpha _xB_y)A_x}=\trace\sqbrac{\rho\trace (\beta _yA_x)B_y}
\end{equation}
so we have
\begin{equation}                
\label{eq46}
\trace (\alpha _xB_y)A_x=\trace (\beta _yA_x)B_y
\end{equation}
Applying \eqref{eq43} and \eqref{eq44} we obtain $\beta _y=\alpha _x=\gamma\in\sscript (H)$ for all $x,y$ and \eqref{eq46} becomes
\begin{equation*}
\trace (\gamma B_y)A_x=\trace (\gamma A_x)B_y
\end{equation*}
for all $x,y$. Summing over $y$ gives $A_x=\trace (\gamma A_x)I_H$. We conclude that if 
\begin{equation}                
\label{eq47}
\hscript ^{(A,\alpha )}\circ\hscript ^{(B,\beta )}=\hscript ^{(B,\beta )}\circ\hscript ^{(A,\alpha )}
\end{equation}
then $AB=BA$. The converse does not hold because we can have $AB=BA$ but \eqref{eq43} does not hold (for example, let $A_x\ne\trace (\gamma A_x)I_H$) so
\eqref {eq47} does not hold.\hfill\qedsymbol
\end{exam}

\section{Measurement Models}  
We begin a study of measurement models \cite{dl70,gud220,hz12}. This section only gives an introduction to the theory and we leave more details to later work. If
$A\in\ob (H)$, we define the \textit{L\"uders instrument} $\lscript ^A\in\instr (H)$ corresponding to $A$ by $\lscript _x^A(\rho )=A_x^{1/2}\rho A_x^{1/2}$ for all $x\in\Omega _A$, $\rho\in\sscript (H)$ \cite{lud51}. Notice that $\lscript ^A$ is a special type of Kraus instrument with Kraus operators $A_x^{1/2}$. Since
\begin{equation*}
\trace\sqbrac{\lscript _x^A(\rho )}=\trace (A_x^{1/2}\rho A_x^{1/2})=\trace (\rho A_x)
\end{equation*}
we have that $(\lscript ^A)^\wedge =A$ and every observable is measured by its corresponding L\"uders instrument. If $A$ is sharp, then $\lscript ^A$ has the form
$\lscript _x^A(\rho )=A_x\rho A_x$.

A measurement model $M$ is an apparatus that can be employed to gain information about a quantum system $S$. If $S$ is described by a Hilbert space $H$, we call $H$ the base space. We interact $H$ with an auxiliary Hilbert space $K$ using an instrument $\iscript\in\instr (H,H\otimes K)$. We then measure a probe observable $P\in\ob (K)$. The result of this measurement gives information about the state of $S$ or observables on $S$. We now make this description mathematically precise.
A \textit{measurement model} is a four-tuple $M=(H,K,\iscript ,P)$ where $H$ is the \textit{base space} Hilbert space, $K$ is the \textit{auxiliary} Hilbert space,
$\iscript\in\instr (H,H\otimes K)$ is the \textit{interaction instrument} and $P\in\ob (K)$ is the \textit{probe observable}. This definition is a generalization of the measurement models that have already been studied \cite{bgl95,hz12}. The \textit{measurement instrument} $\mscript\in\instr (H,H\otimes K)$ for the model $M$ is given by the bi-instrument
\begin{equation*}
\mscript _{xy}=\iscript _y\circ\lscript ^{I_H\otimes P_x}
\end{equation*}
which results from first applying the interaction and then measuring the probe observable. Thus, for all $\rho\in\sscript (H)$ we have
\begin{equation*}
\mscript _{xy}(\rho )=\lscript ^{I_H\otimes P_x}\sqbrac{\iscript _y(\rho )}=(I_H\otimes P_x)^{1/2}\iscript _y(\rho )(I_H\otimes P_x)^{1/2}
\end{equation*}
The measurement instrument contains the information obtained from $M$. In particular, the \textit{marginal measurement instrument} is the instrument
$\mscript ^1\in\instr (H,H\otimes K)$ given by
\begin{equation*}
\mscript _x^1(\rho )=\sum _y\mscript _{xy}(\rho )=\lscript ^{I_H\otimes P_x}\sqbrac{\,\iscriptbar (\rho )}=\iscriptbar\circ\lscript ^{I_H\otimes P_x}(\rho )
\end{equation*}
We call the reduced marginal instrument $\mscript _1^1\in\instr (H)$ the \textit{instrument measured by} $M$ and we obtain
\begin{equation*}
\mscript _{1x}^1(\rho )=\trace _K\sqbrac{\mscript _x^1(\rho )}=\trace _K\sqbrac{\,\iscriptbar\circ\lscript ^{I_H\otimes P_x}(\rho )}
\end{equation*}
for all $\rho\in\sscript (H)$. We call the observable $\mscripthat\colon\ob (H)$, the \textit{observable measured by} $M$. Since $\mscripthat _1^{\,1}$ satisfies
\begin{align*}
\trace (\rho\mscripthat _{1x}^{\,1})&=\trace\sqbrac{\mscript _{1x}^1(\rho )}=\trace\sqbrac{\lscript ^{I_H\otimes P_x}(\,\iscriptbar (\rho ))}\\
   &=\trace\sqbrac{(I_H\otimes P_x)^{1/2}\iscriptbar\,(\rho )(I_H\otimes P_x)^{1/2}}\\
   &=\trace\sqbrac{\,\iscriptbar (\rho )(I_H\otimes P_x)}=\trace\sqbrac{\rho\iscriptbar\,^*(I_H\otimes P_x)}
\end{align*}
we conclude that $\mscripthat _{1x}^{\,1}=\iscriptbar\,^*(I_H\otimes P_x)$.

Suppose $\iscript$ is a Holevo instrument $\iscript =\hscript ^{(A,\alpha )}$, where $A\in\ob (H)$ and
$\alpha =\brac{\alpha _x\colon x\in\Omega _A}\subseteq\sscript (H\otimes K)$. Then
\begin{align*}
\mscript _{xy}(\rho )&=\lscript ^{I_H\otimes P_x}(\iscript _y(\rho ))=\lscript ^{I_H\otimes P_x}\sqbrac{\trace (\rho A_y)\alpha _y}
   =\trace (\rho A_y)\lscript ^{I_H\otimes P_x}(\alpha _y)\\
   &=\trace (\rho A_x)(I_H\otimes P_x)^{1/2}\alpha _y(I_H\otimes P_x)^{1/2}
\end{align*}
and we obtain the instrument measured by $M$:
\begin{equation*}
\mscript _{1x}^1(\rho )=\sum _y\trace (\rho A_y)\trace _K\sqbrac{(I_H\otimes P_x)^{1/2}\alpha _y(I_H\otimes P_x)^{1/2}}
\end{equation*}
Since $\iscript _y^*(a)=\trace (\alpha _ya)A_y$ for all $a\in\escript (H\otimes K)$ we have $\iscriptbar\,^*(a)=\sum _y\trace (\alpha _ya)A_y$. Then the observable measured by $M$ becomes
\begin{equation*}
\mscripthat _{1x}^{\,^1}=\iscriptbar\,^*(I_H\otimes P_x)=\sum _y\trace\sqbrac{\alpha _y(I_H\otimes P_x)}A_y
\end{equation*}
which is a post-processing of $A$ because $\sum _x\trace\sqbrac{\alpha _y(I_H\otimes P_x)}=1$ for all $y$. In the particular case where $P$ is sharp and
$\alpha _y=\beta _y\otimes\gamma _y$, $\beta _y\in\sscript (H)$, $\gamma _y\in\sscript (K)$ we obtain
\begin{equation*}
\mscript _{xy}(\rho )=\trace (\rho A_y)(I_H\otimes P_x)\beta _y\otimes\gamma _y(I_H\otimes P_x)=\trace (\rho A_y)\beta _y\otimes P_x\gamma _yP_x
\end{equation*}
It follows that
\begin{align*}
\mscript _{1x}^1(\rho )&=\sum _y\trace (\rho A_y)\trace _K(\beta _y\otimes P_x\gamma _yP_x)=\sum _y\trace (\rho A_y)\trace (P_x\gamma _y)\beta _y\\
   &=\sum _y\trace\sqbrac{P_x\hscript _y^{(A,\gamma )}(\rho )}\beta _y\\
\intertext{and}
\mscripthat _{1x}^{\,1}&=\sum _y\trace\sqbrac{\beta _y\otimes\gamma _y(I_H\otimes P_x)}A_y\\
   &=\sum _y\trace (\beta _y\otimes\gamma _yP_x)A_x=\sum _y\trace (\gamma _yP_x)A_y
\end{align*}

Finally, we introduce the sequential product of measurement models. Let $M=(H,K,\iscript ,P)$ and $M_1=(H\otimes K,K_1,\iscript _1,P_1)$ be measurement models where $\iscript\in\instr (H,H\otimes K)$, $P\in\ob (K)$, $\iscript _1\in\instr (H\otimes K,H\otimes K\otimes K_1)$, $P_1\in\ob (K_1)$. The
\textit{sequential product of} $M$ \textit{then} $M_1$ is the measurement model
\begin{equation*}
M_2=M\circ M_1=(H,K\otimes K_1,\iscript _2,P_2)
\end{equation*}
where $\iscript _2\in\instr (H,H\otimes K\otimes K_1)$ is given by $\iscript _2=\iscript\circ\iscript _1$ and $P_2\in\ob (K\otimes K_1)$ is given by
$P_{2xy}=P_x\otimes P_{1y}$. The corresponding measurement instrument for $M_2$ because the 4-instrument $\mscript\in\instr (H,H\otimes K\otimes K_1)$ defined as
\begin{equation*}
\mscript _{xyx'y'}=\iscript _{2x'y'}\circ\lscript ^{I_H\otimes P_{2xy}}
\end{equation*}
Hence,
\begin{align*}
\mscript _{xyx'y'}(\rho )&=\lscript ^{I_H\otimes P_{2xy}}\sqbrac{\iscript _{2x'y'}(\rho )}=(I_H\otimes P_{2xy})^{1/2}\iscript _{2x'y}(\rho )(I_H\otimes P_{2xy})^{1/2}\\
   &=(I_H\otimes P_x\otimes P_{1y})^{1/2}\iscript _{1x'}(\iscript _{y'}(\rho ))(I_H\otimes P_x\otimes P_{1y})^{1/2}
\end{align*}
The marginal measurement $\mscript _{xy}^1\in\instr (H,H\otimes K\otimes K_1)$ becomes
\begin{equation*}
\mscript _{xy}^1(\rho )=\sum _{x',y'}\mscript _{xyx'y'}=(I_H\otimes P_x\otimes P_{1y})^{1/2}\iscriptbar _1(\,\iscriptbar (\rho ))(I_H\otimes P_x\otimes P_{1y})^{1/2}
\end{equation*}
We then obtain the instrument $\mscript _{1xy}^1\in\instr (H)$ measured by $M_2$ as
\begin{equation*}
\mscript _{1xy}^1(\rho )=\trace _{K\otimes K_1}\sqbrac{\mscript _{xy}^1(\rho )}
\end{equation*}
and the observable $\mscripthat _{1xy}^{\,1}$ measured $M_2$ becomes
\begin{equation*}
\mscripthat _{1xy}^{\,1}=\iscriptbar\,^*\sqbrac{\,\iscriptbar _1^{\,*}(I_H\otimes P_x\otimes P_{1y})}
\end{equation*}

\end{document}